\documentclass[aps,10pt,onecolumn,prx,longbibliography,nofootinbib,superscriptaddress]{revtex4-2}
\usepackage[T1]{fontenc}
\usepackage{times}
\usepackage[a4paper,margin=1in]{geometry}
\usepackage{amsmath,amssymb,mathtools,amsthm}
\usepackage{booktabs}
\usepackage{xcolor}
\usepackage{pgfplots}
\pgfplotsset{compat=1.18}
\usepackage{microtype}
\usepackage[colorlinks=true,linkcolor=teal,citecolor=teal,urlcolor=teal]{hyperref}
\providecommand{\doi}[1]{\href{https://doi.org/#1}{doi:\nolinkurl{#1}}}
\hypersetup{pdftitle={Canonical purification and nonadditivity in entanglement of purification},
 pdfauthor={Zahra Baghali Khanian and Amir-Reza Negari}}
\usepackage{enumitem}
\usepackage{needspace}
\usepackage{placeins}
\setlist{itemsep=2pt,topsep=4pt}
\newcommand{\Tr}{\operatorname{Tr}}
\newcommand{\EP}{E_{\mathrm P}}
\newcommand{\EPr}{E_{\mathrm P}^{\infty}}
\newcommand{\ket}[1]{\lvert #1\rangle}
\newcommand{\bra}[1]{\langle #1\rvert}
\newcommand{\proj}[1]{\ket{#1}\!\bra{#1}}
\newcommand{\id}{I}
\newcommand{\Ps}{P_{\mathrm{sym}}}
\newcommand{\norm}[1]{\lVert #1\rVert}
\newcommand{\hbin}{h_2}

\newtheorem{theorem}{Theorem}[section]
\newtheorem{lemma}[theorem]{Lemma}
\newtheorem{proposition}[theorem]{Proposition}
\newtheorem{appendixlemma}{Lemma}[section]
\allowdisplaybreaks[2]
\makeatletter
\let\auto@bib\@empty
\let\auto@bib@innerbib\@empty
\makeatother

\begin{document}

\title{Entanglement of purification for Werner states: canonical purification and nonadditivity}

\author{Zahra Baghali Khanian}
\affiliation{Perimeter Institute for Theoretical Physics,
Waterloo, Ontario N2L 2Y5, Canada}
\affiliation{Institute for Quantum Computing, University of Waterloo,
Waterloo, Ontario N2L 3G1, Canada}

\author{Amir-Reza Negari}
\affiliation{Institute for Quantum Computing, University of Waterloo,
Waterloo, Ontario N2L 3G1, Canada}

\begin{abstract}
We prove that the von Neumann entanglement of purification is
nonadditive. For two-qubit Werner states $W(f)$ with singlet
weight $1/64\le f\le1/4$, we show that canonical purification is
globally optimal: its entanglement, also known as reflected
entropy, gives the exact entanglement of purification. In
contrast, the regularized entanglement of purification lies
strictly below the reflected entropy for $0<f<f_\ast$, where
$f_\ast\approx0.017$. Together, these results establish
nonadditivity throughout $1/64\le f<f_\ast$. Our
canonical-optimality proof generalizes the proof in \cite{NegariKhanian}, which established R\'enyi
nonadditivity for $0\le\alpha<1$ using classical two-qubit
states. We thus resolve the additivity question at
$\alpha=1$. Remarkably, nonadditivity is again witnessed
by separable states and no entanglement in the underlying states
is required.
\end{abstract}

\maketitle
\raggedbottom

\begingroup
\setlength{\parskip}{2pt plus 0.4pt minus 0.4pt}

\section{Introduction}\label{sec:introduction}
\clubpenalty=10000
\widowpenalty=10000

The resources needed to prepare a correlated quantum state are not
captured by its entanglement alone. Entanglement of purification,
introduced by Terhal, Horodecki, Leung, and DiVincenzo~\cite{THLD},
quantifies total correlations by minimizing the entanglement of a
purification across the cut that groups each party with its own
ancillary system. This perspective on preparation complements the
characterization of total correlations through the noise required to
erase them~\cite{GroismanPopescuWinterCorrelations}. The regularized
entanglement of purification is the asymptotic entanglement cost of
state preparation under local operations with communication
sublinear in the number of copies~\cite{THLD}. The additivity question
is whether optimizing one copy already gives the cost per
copy, or whether a joint purification can offer a collective saving.
Here we prove such a saving for a continuous interval of separable
two-qubit Werner states.

The importance of regularization extends beyond this preparation
problem. The regularized entanglement of purification appears in the
optimal visible-compression rate of
mixed-state ensembles~\cite{HayashiVisibleCompression} and in
unassisted quantum rate-distortion
functions~\cite{DattaHsiehWildeRateDistortion,BaghaliKhanianKuroiwaLeungRateDistortion}.
R\'enyi versions of entanglement of purification also enter strong
converse bounds for visible compression~\cite{BaghaliKhanianStrongConverse},
while the single-copy quantity is related to quantum advantage in
dense coding~\cite{HorodeckiPianiDenseCoding}. These connections make
the relation between the single-copy and asymptotic quantities a
basic question about the operational meaning of quantum correlations.

Despite these operational roles, entanglement of purification is
difficult to evaluate. Bounds and exact results are available for
selected state families~\cite{ChristandlWinter,BagchiPatiProperties},
but a useful purification generally supplies only an upper bound.
This difficulty also matters in many-body physics: optimizing the
purifying degrees of freedom can reduce the entanglement needed in
matrix-product representations of thermal and time-evolved
states~\cite{HauschildPurifications,PokartPurifications}, and suitable
bounds on R\'enyi entanglement of purification give rigorous criteria
for efficient matrix-product density-operator
approximations~\cite{JarkovskyMPDO}. Gaussian approaches have enabled
calculations in free scalar theories and free conformal field
theories~\cite{BhattacharyyaFreeScalar,CamargoFreeCFT,WindtGaussian}.
Numerical studies also show that optimized purifications can break
the symmetry exchanging the two physical
subsystems~\cite{BhattacharyyaSymmetryBreaking}. Symmetry of the target
state alone therefore does not establish optimality of a symmetric
purification.

A distinguished candidate is the canonical purification, obtained by
vectorizing the positive square root of the density operator. Its
entanglement is the reflected entropy~\cite{DuttaFaulkner}, which
requires no optimization and is always an upper bound on
entanglement of purification. Both quantities have attracted interest
in quantum field theory and holography. In holographic theories,
entanglement of purification was conjectured to equal the entanglement wedge
cross-section~\cite{UmemotoTakayanagi,NguyenSpinHolography}, whereas
reflected entropy gives twice that cross-section at leading
semiclassical order~\cite{DuttaFaulkner}. The former identification
has been established for a class of random tensor networks at large
bond dimension~\cite{AkersEPRTN}. For widely separated spherical
regions in a conformal-field-theory vacuum, both quantities exhibit a
logarithmic enhancement relative to mutual
information~\cite{CamargoLongDistance}. These connections do not make
the two quantities interchangeable: reflected entropy can increase
under partial trace~\cite{HaydenLemmSorce}, and even one-half of it is
not a universal lower bound on entanglement of
purification~\cite{CouchEtAl}. Identifying states for which the
canonical upper bound is attained is therefore a substantive
optimization problem.

Two-qubit Werner states~\cite{Werner} bring this problem into a simple
one-parameter setting. The original study of entanglement of
purification found numerical evidence that the canonical
purification is optimal over an interval of singlet weights ending
at the maximally mixed state~\cite{THLD}. Chen and
Winter~\cite{ChenWinter} subsequently proved that the regularized
entanglement of purification minus the state entropy is convex.
Their resulting ensemble bound, evaluated at singlet weight
$f=0.005$, lies below the numerically estimated single-copy value and
provides strong evidence for nonadditivity. To turn this strategy
into an analytical separation, one needs a rigorous single-copy
lower bound that exceeds an achievable asymptotic rate. An AI-generated proof of nonadditivity using this principle was presented in \cite{Krohn} and its accompanying artifact~\cite{KrohnArtifact}. We show that this proof is incorrect and provide an explicit counterexample to Lemma~2, invalidating the single-copy lower-bound argument on which the proof rests. 

\Needspace{6\baselineskip}
In this paper, we establish such a separation. For the two-qubit
Werner state $W(f)$ with singlet weight $f$, we prove that the
canonical purification is optimal throughout $1/64\le f\le1/4$.
Thus its reflected entropy gives the exact entanglement of
purification on this interval. We also show that the regularized
quantity lies strictly below the canonical value for $0<f<f_\ast$,
where $f_\ast\approx0.017086$. Combining the two results gives
\begin{equation}\label{eq:intro-nonadditivity}
  \EPr(W(f))<\EP(W(f)),
  \qquad 1/64\le f<f_\ast.
\end{equation}
All states in this interval are separable. Collective preparation
therefore reduces the entanglement cost even when the target states
themselves contain no entanglement. Figure~\ref{fig:werner-intervals}
summarizes the two bounds and the interval on which they prove
nonadditivity.

The main technical step is the single-copy lower bound. Projecting an
arbitrary purification onto the symmetric subspace of the physical
qubits produces a state constrained by symmetric extendibility.
The two-qubit extension criterion~\cite{CJKLZ}, together with the
overlap between the original and projected purifications, yields a
joint constraint on the largest Schmidt probability and the
remaining probabilities. Grouping the tail reduces the entropy
comparison to four probabilities. We prove that the canonical
spectrum minimizes this entropy throughout $1/64\le f\le1/4$.
The asymptotic improvement then follows from the ensemble bound of
Chen and Winter~\cite{ChenWinter}.

\par
\endgroup

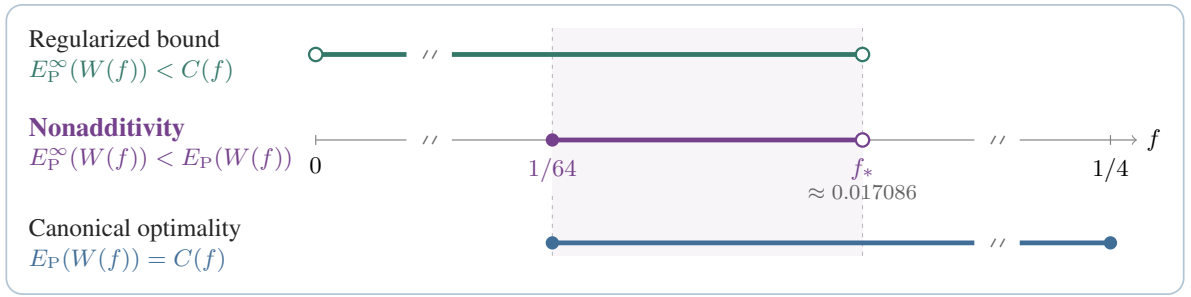
\begin{figure}[!htb]
\centering
\begin{tikzpicture}[x=1cm,y=1cm,
  font=\fontsize{9.5}{11}\selectfont,
  line cap=butt,
  bound/.style={line width=1.8pt},
  endpoint/.style={line width=0.85pt},
  axis/.style={draw=black!48,line width=0.45pt}]
  \definecolor{WernerBlue}{HTML}{416E97}
  \definecolor{WernerTeal}{HTML}{34786A}
  \definecolor{WernerPurple}{HTML}{75418C}
  \pgfmathsetmacro{\plotStart}{3.80}
  \pgfmathsetmacro{\previousWidth}{4000*(0.01708595559070344-1/64)}
  \pgfmathsetmacro{\overlapCenter}{\plotStart+2.25+0.5*\previousWidth}
  \pgfmathsetmacro{\fzeroX}{\overlapCenter-0.35*\previousWidth}
  \pgfmathsetmacro{\fstarX}{\overlapCenter+0.35*\previousWidth}
  \pgfmathsetmacro{\leftEnd}{\plotStart+1.2}
  \pgfmathsetmacro{\middleStart}{\plotStart+1.8}
  \pgfmathsetmacro{\middleEnd}{\plotStart+2.25+\previousWidth+0.6}
  \pgfmathsetmacro{\rightStart}{\middleEnd+0.6}
  \pgfmathsetmacro{\quarterX}{\rightStart+1.2}

  \fill[WernerPurple!5] (\fzeroX,-1.53) rectangle (\fstarX,1.48);
  \foreach \xx in {\fzeroX,\fstarX}{
    \draw[black!28,line width=0.4pt,dash pattern=on 1.5pt off 2pt]
      (\xx,-1.53) -- (\xx,-0.72)
      (\xx,0.08) -- (\xx,1.48);
  }

  \node[anchor=west,inner sep=0pt,text=black!85] at (0,1.30)
    {Regularized bound};
  \node[anchor=west,inner sep=0pt,text=WernerTeal,
        font=\fontsize{9}{11}\selectfont] at (0,0.91)
    {$\EPr(W(f))<C(f)$};
  \node[anchor=west,inner sep=0pt,text=WernerPurple,
        font=\fontsize{10}{12}\selectfont\bfseries] at (0,0.14)
    {Nonadditivity};
  \node[anchor=west,inner sep=0pt,text=WernerPurple,
        font=\fontsize{9}{11}\selectfont] at (0,-0.25)
    {$\EPr(W(f))<\EP(W(f))$};
  \node[anchor=west,inner sep=0pt,text=black!85] at (0,-1.19)
    {Canonical optimality};
  \node[anchor=west,inner sep=0pt,text=WernerBlue,
        font=\fontsize{9}{11}\selectfont] at (0,-1.58)
    {$\EP(W(f))=C(f)$};

  \draw[WernerTeal,bound]
    (\plotStart,1.13) -- (\leftEnd,1.13)
    (\middleStart,1.13) -- (\fstarX,1.13);
  \draw[WernerTeal,fill=white,endpoint]
    (\plotStart,1.13) circle[radius=2.5pt]
    (\fstarX,1.13) circle[radius=2.5pt];
  \draw[WernerBlue,bound]
    (\fzeroX,-1.36) -- (\middleEnd,-1.36)
    (\rightStart,-1.36) -- (\quarterX,-1.36);
  \fill[WernerBlue] (\fzeroX,-1.36) circle[radius=2.5pt];
  \fill[WernerBlue] (\quarterX,-1.36) circle[radius=2.5pt];

  \draw[axis] (\plotStart,0) -- (\leftEnd,0);
  \draw[axis] (\middleStart,0) -- (\middleEnd,0);
  \draw[axis,->] (\rightStart,0) -- ({\quarterX+0.35},0);
  \draw[WernerPurple,bound] (\fzeroX,0) -- (\fstarX,0);
  \fill[WernerPurple] (\fzeroX,0) circle[radius=2.5pt];
  \draw[WernerPurple,fill=white,endpoint]
    (\fstarX,0) circle[radius=2.5pt];

  \foreach \xx/\yy in {\leftEnd/1.13,\leftEnd/0,\middleEnd/0,\middleEnd/-1.36}{
    \draw[black!52,line width=0.6pt]
      ({\xx+0.22},{\yy-0.055}) -- ({\xx+0.28},{\yy+0.055})
      ({\xx+0.34},{\yy-0.055}) -- ({\xx+0.40},{\yy+0.055});
  }
  \foreach \xx in {\plotStart,\quarterX}{
    \draw[axis] (\xx,0.065) -- (\xx,-0.065);
  }

  \node[anchor=north,inner sep=0pt] at (\plotStart,-0.23) {$0$};
  \node[anchor=north,inner sep=0pt,text=WernerPurple]
    at (\fzeroX,-0.23) {$1/64$};
  \node[anchor=north,inner sep=0pt,text=WernerPurple]
    at (\fstarX,-0.23) {$f_\ast$};
  \node[anchor=north,inner sep=0pt,text=black!62,
        font=\fontsize{8.5}{10}\selectfont]
    at (\fstarX,-0.59) {$\approx0.017086$};
  \node[anchor=north,inner sep=0pt] at (\quarterX,-0.23) {$1/4$};
  \node[anchor=west,inner sep=0pt] at ({\quarterX+0.47},0) {$f$};

  \draw[WernerBlue!40,rounded corners=6pt,line width=0.65pt]
    ([xshift=-8.5pt,yshift=-8.5pt]current bounding box.south west)
    rectangle
    ([xshift=8.5pt,yshift=8.5pt]current bounding box.north east);

\end{tikzpicture}
\caption{Canonical optimality and nonadditivity for two-qubit Werner
states. The green interval gives $\EPr(W(f))<C(f)$ for $0<f<f_\ast$
(Proposition~\ref{prop:upper}), where $C(f)$ denotes the entanglement entropy of the canonical purification
of $W(f)$ across the bipartition $AA':BB'$, equivalently the reflected
entropy of $W(f)$.  The blue interval gives
$\EP(W(f))=C(f)$ for $1/64\le f\le1/4$
(Proposition~\ref{prop:overlap-optimum}). Their intersection, shown in
purple, proves nonadditivity on $[1/64,f_\ast)$.
Filled endpoints are included and open endpoints excluded.
The axis is broken on both sides to enlarge the overlap;
unmarked intervals are not classified by these results.}
\label{fig:werner-intervals}
\end{figure}
\FloatBarrier

Our result complements the R\'enyi nonadditivity theorem of
Ref.~\cite{NegariKhanian}. That work establishes canonical
optimality for classical two-qubit states at every R\'enyi order and
exhibits two-copy violations for $0\le\alpha<1$. Here we establish
nonadditivity at the von Neumann order $\alpha=1$ using separable
Werner states. These states are not diagonal in a local product
basis, so the corresponding von Neumann question for classical
states remains open. At R\'enyi order $2$, additivity has been
established for families associated with depolarizing and
transpose-depolarizing channels~\cite{FarajiKhanian}.

\medskip
After the preliminaries in Section~\ref{sec:result}, we establish the
regularized bound in Section~\ref{sec:upper} and prove canonical
optimality in Section~\ref{sec:overlap}.
Section~\ref{sec:discussion} discusses the implications and open
questions. The appendices contain the detailed calculations.

\section{Preliminaries}\label{sec:result}

All systems are finite dimensional, and entropies are measured in
bits. Logarithms are to base two unless written as $\ln$, with
$0\log_2 0=0$. The von Neumann, Shannon, and binary entropies are
\[
 \textstyle
 S(\rho)=-\Tr\!\left(\rho\log_2\rho\right),
 \qquad H(p)=-\sum_i p_i\log_2 p_i,
 \qquad \hbin(x)=H(x,1-x).
\]

\subsection{Purification and regularization}

A purification of $\rho_{AB}$ is a pure state
$\ket\psi_{AA'BB'}$ with $\Tr_{A'B'}\proj\psi=\rho_{AB}$,
where $A'$ and $B'$ are local ancillary systems. Its entanglement
across $AA':BB'$ is $S(\Tr_{BB'}\proj\psi)=H(p)$, where $p$ is
its vector of Schmidt probabilities. Minimizing this cost over all
purifications defines the entanglement of purification\cite{THLD}:
\begin{equation}\label{eq:ep}
 \EP(\rho_{AB})=\inf_{\psi}
 S\!\left(\Tr_{BB'}\proj\psi\right).
\end{equation}
Arbitrary finite ancillary dimensions are allowed. For
$d_A=\dim A$ and $d_B=\dim B$, the infimum is attained with
$\dim A'\le d_A d_B$ and $\dim B'\le(d_A d_B)^2$, or with the
ancillas interchanged\cite[Lemma~1]{THLD}. Thus dimensions
$4$ and $16$ suffice for two-qubit states.

Tensoring purifications shows that $\EP$ is subadditive. Hence its
regularization exists and satisfies
\[
 \EPr(\rho)
 =\lim_{n\to\infty}\frac{\EP(\rho^{\otimes n})}{n}
 =\inf_{n\ge1}\frac{\EP(\rho^{\otimes n})}{n}
 \le\EP(\rho).
\]
Strict inequality holds if and only if
$\EP(\rho^{\otimes n})<n\EP(\rho)$ for some finite $n\ge2$.

Vectorizing $\sqrt\rho$ in a product basis gives the canonical
purification $\ket{\sqrt\rho}$, with reference copies $A'$ and $B'$
assigned to the corresponding parties. Its entanglement across
$AA':BB'$ is the reflected entropy $S_R(A:B)_\rho$\cite{DuttaFaulkner}.
Thus $\EP(\rho)\le S_R(A:B)_\rho$, and canonical optimality means
equality in this bound.

\subsection{Werner states and the canonical spectrum}

The two-qubit Werner family\cite{Werner} is parametrized by its
singlet probability:
\begin{equation}\label{eq:werner}
 W(f)=f\proj{\Psi^-}+\frac{1-f}{3}\Ps,
 \qquad 0\le f\le1.
\end{equation}
Here $\ket{\Psi^-}=(\ket{01}-\ket{10})/\sqrt2$ and
$\Ps=\id_{AB}-\proj{\Psi^-}$ projects onto the symmetric subspace.
Both marginals are $\id/2$, and $W(f)$ is separable exactly when
$f\le1/2$. We consider $0\le f\le1/4$, from the normalized
symmetric projector $W(0)=\Ps/3$ to the product state
$W(1/4)=\id_{AB}/4$. Its state entropy is
$S(W(f))=\hbin(f)+(1-f)\log_2 3$.

For $0\le f\le1/4$, the canonical purification has Schmidt
amplitudes $(a_f,b_f,b_f,b_f)$ across $AA':BB'$, where
\begin{equation}\label{eq:overlap-canonical-amplitudes}
 a_f=\frac{\sqrt f+\sqrt{3(1-f)}}2,
 \qquad
 b_f=\frac{\sqrt{(1-f)/3}-\sqrt f}{2}.
\end{equation}
We denote its entanglement by
\begin{equation}\label{eq:canonical-entropy}
 C(f)=S_R(A:B)_{W(f)}
 =H(a_f^2,b_f^2,b_f^2,b_f^2),
\end{equation}
and write $q(f)=a_f^2$ for its largest Schmidt probability.
Appendix~\ref{app:bell} gives the explicit purification and
its Schmidt decomposition.

\section{A regularized cost below the canonical value}
\label{sec:upper}

Collective preparation can require less entanglement per copy than
canonical purification. To show this, we use an ensemble decomposition
of the target state. For any $\rho=\sum_i p_i\rho_i$, we have \cite{ChenWinter}
\begin{equation}\label{eq:ensemble-regularized-bound}
 \EPr(\rho)
 \le
 \sum_i p_i\EPr(\rho_i)
 +S(\rho)-\sum_i p_iS(\rho_i).
\end{equation}
The right-hand side is an achievable asymptotic entanglement rate.
Its first term is the average entanglement cost of the component
states. The entropy difference is the Holevo information of the
ensemble and contributes an additional rate in this preparation
scheme. The choice of ensemble balances these two contributions.

For the Werner family, we use $W(0)$ as one component because its
optimal purification requires less entanglement than its canonical
purification:
\begin{equation}\label{eq:cheap-endpoint}
 \EP(W(0))=1<C(0).
\end{equation}
The equality is recalled in Appendix~\ref{app:wzero}. Fix
$0<f<1/4$ and assign probability $p$ to $W(0)$. To keep the target
state unchanged, increase the singlet probability of the other
component:
\begin{equation}\label{eq:variable-endpoint-mixture}
 W(f)=p\,W(0)+(1-p)\,W\!\left(\frac{f}{1-p}\right),
 \qquad
 0\le p\le 1-4f.
\end{equation}
This identity follows from the affine dependence of $W(f)$ on $f$.
As $p$ increases, the second component moves toward the product
state $W(1/4)$, and its canonical cost decreases.

Using one ebit to prepare $W(0)$ and canonical purification
for the second component in~\eqref{eq:ensemble-regularized-bound}
gives
\begin{equation}\label{eq:variable-endpoint-cost}
 \EPr(W(f))\le V_f(p)
 :=p+(1-p)C\!\left(\frac{f}{1-p}\right)+S(W(f))-p\log_2 3
 -(1-p)S\!\left(W\!\left(\frac{f}{1-p}\right)\right),
\end{equation}
where $S(W(0))=\log_2 3$. The first line is the
average preparation rate; the second is the Holevo information.

At $p=0$, the ensemble contains only $W(f)$, so the Holevo information
vanishes and $V_f(0)=C(f)$. A negative right derivative therefore
guarantees a lower achievable rate for sufficiently small $p>0$.
Differentiation gives
\begin{equation}\label{eq:initial-endpoint-derivative}
 V_f'(0^+)
 =
 \log_2\frac83
 -2\hbin\!\left(
 \frac{1-\sqrt{3f/(1-f)}}4
 \right).
\end{equation}
The right-hand side is strictly increasing in $f$ and has a unique
zero $f_\ast\approx0.017086$ on $(0,1/4)$. For $0<f<f_\ast$, it is
negative: the decrease in the average preparation rate exceeds the
additional Holevo information. Appendix~\ref{app:upper-proof} gives
the calculation and sign analysis, establishing the following result.

\begin{proposition}\label{prop:upper}
For every $0<f<f_\ast$, all sufficiently small $p>0$ satisfy
\begin{equation}\label{eq:strict-canonical-upper}
 \EPr(W(f))\le V_f(p)<C(f).
\end{equation}
\end{proposition}

\begin{proof}
For $0<f<f_\ast$, we have $V_f(0)=C(f)$ and
\[
 \lim_{p\downarrow0}\frac{V_f(p)-C(f)}{p}=V_f'(0^+)<0.
\]
Thus $V_f(p)<C(f)$ for all sufficiently small $p>0$.
Combining this with~\eqref{eq:variable-endpoint-cost} gives the claim.
\end{proof}

This bound on the asymptotic rate is independent of single-copy
optimality. Proposition~\ref{prop:overlap-optimum} proves
$\EP(W(f))=C(f)$ for $1/64\le f\le1/4$. Combining the two results
gives $\EPr(W(f))<\EP(W(f))$ for $1/64\le f<f_\ast$, proving
nonadditivity on this interval.

\section{Optimality of the canonical purification}\label{sec:overlap}

We prove that the canonical purification attains the single-copy
entanglement of purification of $W(f)$ for $1/64\le f\le1/4$.

In this section, we extend the largest-eigenvalue approach of
 \cite{NegariKhanian} to establish canonical optimality for Werner
states. The quantity to compare is the largest eigenvalue of the
reduced purification state $\Tr_{BB'}\proj\psi$, equivalently its
largest Schmidt probability $p_1$ across $AA':BB'$. For the classical
two-qubit states considered in  \cite{NegariKhanian}, the canonical
spectrum has at most two nonzero entries, $(\lambda_+,\lambda_-)$
with $\lambda_+\ge\lambda_-$. The bound $p_1\le\lambda_+$ then shows
that the canonical spectrum majorizes every competitor: its first
entry is at least as large, and its first two entries already sum to
one. This proves canonical optimality. For Werner states, the canonical
purification likewise maximizes the largest eigenvalue of the reduced
state on $AA'$: every purification satisfies $p_1\le a_f^2$,
and the canonical purification attains equality, as proved in
Lemma~\ref{lem:overlap-leading} in Appendix~\ref{app:overlap-leading}.
However, its spectrum has four nonzero entries when $f<1/4$. A
competitor with a smaller largest eigenvalue could concentrate the
remaining probability into fewer entries and potentially lower its
entropy. Our extension controls the leading probability jointly with
the tail, providing the additional information needed to prove the
entropy lower bound. Appendix~\ref{app:overlap-majorization} explains
why a bound on the largest eigenvalue alone cannot establish a
universal majorization comparison in this case.

The proof of the von Neumann entropy bound has two steps. First,
Lemma~\ref{lem:overlap-tail} constrains the leading probability
jointly with the tail. Projecting onto the symmetric subspace
produces an extendible marginal, and the large overlap with this
projection transfers its spectral restriction to the original
purification. Second, Lemma~\ref{lem:overlap-entropy} minimizes the
entropy of the resulting four-entry probability vector. Grouping
the tail cannot increase entropy, so this minimum also bounds the
cost of every purification. When $p_1>1/2$, reducing it increases
the binary contribution $\hbin(p_1)$; the entropy minimum shows
that concentrating the tail cannot save more than this increase
costs. The case $p_1\le1/2$ already has entropy at least one bit.
The canonical spectrum attains the resulting lower bound.

The canonical amplitudes in
Eq.~\eqref{eq:overlap-canonical-amplitudes} satisfy
\begin{equation}\label{eq:overlap-canonical-relations}
 a_f^2+3b_f^2=1,
 \qquad a_f-3b_f=2\sqrt f,
 \qquad a_f\ge b_f\ge0.
\end{equation}
These identities connect the canonical spectrum to the
constraints used below. For an arbitrary purification,
let $p_1\ge p_2\ge\cdots$ denote its Schmidt probabilities
across $AA':BB'$. Its entanglement is $H(p)$, and our task
is to prove $H(p)\ge C(f)$ for $1/64\le f\le1/4$,
without restricting the finite ancillary dimensions. Throughout this section,
$p_1\ge p_2\ge\cdots$ denotes the Schmidt probabilities of an
arbitrary purification across the same cut $AA':BB'$. The ancillary
dimensions are unrestricted apart from being finite.
Equivalently, the $p_i$ are the eigenvalues of
$\Tr_{BB'}\proj\psi$, the Schmidt coefficients are $\sqrt{p_i}$,
and the purification cost is $S(AA')_\psi=H(p)$.

The spectral step uses the two-qubit extension criterion of
 \cite{CJKLZ}. The Werner-state form of the resulting tail
inequality also appears in  \cite{Krohn}; its derivation here is
given directly from the criterion and elementary overlap estimates.
Appendix~\ref{app:overlap} supplies the detailed calculations for
all steps, including the entropy minimization.

\subsection{An overlap bound for the spectral tail}

The next lemma constrains the leading probability jointly with the
tail. It applies to every two-qubit state with maximally mixed marginals,
with $f$ denoting its probability of the physical singlet outcome.
We say that $\sigma_{XY}$ has a symmetric extension in $Y$ if it
is the $XY$ marginal of a state on $XYY'$ invariant under exchange
of $Y$ and $Y'$. In the proof, exchange symmetry of the physical
qubits produces such an extension with $X=A'$, $Y=A$, and $Y'=B$.

\begin{lemma}\label{lem:overlap-tail}
Let $\rho_{AB}$ be a two-qubit state with
$\rho_A=\rho_B=\id/2$, and set
$f=\bra{\Psi^-}\rho_{AB}\ket{\Psi^-}<1$.
For every finite-dimensional purification of $\rho_{AB}$, let
$p_1\ge p_2\ge\cdots$ be its Schmidt probabilities across
$AA':BB'$, padded with zeros, and $Q=\sum_{i\ge4}p_i$. Then
\begin{equation}\label{eq:overlap-direct-tail}
 \sqrt{p_1}-\sqrt{p_2}-\sqrt{p_3}-\sqrt Q\le2\sqrt f.
\end{equation}
\end{lemma}

\begin{proof}
Project a purification $\ket\psi$ onto the physical symmetric subspace
and normalize it:
\begin{equation}\label{eq:overlap-symmetric-projection}
 \ket\chi=\frac{(\Ps\otimes\id_{A'B'})\ket\psi}{\sqrt{1-f}},
 \qquad \langle\psi|\chi\rangle=\sqrt{1-f}.
\end{equation}
The projected state remains maximally mixed on each physical qubit
and is invariant under exchange of $A$ and $B$. 
Indeed, writing $V_{AB}$ for the exchange operator, we have
$V_{AB}\Ps=\Ps$, while taking either partial trace of
$2\Ps\rho_{AB}\Ps
=\rho_{AB}+V_{AB}\rho_{AB}V_{AB}
-2f\proj{\Psi^-}$
gives $(1-f)\id$, so normalization yields the marginal
$\id/2$ on each physical qubit.
Thus its $A'A$
marginal has the symmetric extension $\Tr_{B'}\proj\chi$ on $A'AB$,
which is invariant under exchanging $A$ and $B$, with $A'$ unchanged.
After the two-dimensional compression described in
Appendix~\ref{app:overlap-halfspace}, we apply the two-qubit
symmetric-extension criterion of \cite{CJKLZ}, stated
in~\eqref{eq:overlap-compressed-criterion}; factoring the resulting
polynomial gives
\begin{equation}\label{eq:overlap-endpoint-halfspace}
 t_1\le t_2+t_3+\left(\sum_{i\ge4}t_i^2\right)^{1/2},
\end{equation}
where $t_i$ are the ordered Schmidt coefficients of $\ket\chi$
across $AA':BB'$. Appendix~\ref{app:overlap-halfspace} derives this
constraint directly from the criterion, including the compression
and algebra.

We now use the overlap in~\eqref{eq:overlap-symmetric-projection}
to transfer this constraint from $\ket\chi$ to $\ket\psi$.
Put $s_i=\sqrt{p_i}$ and group the remaining Schmidt coefficients
into one entry in each vector:
\begin{equation}\label{eq:overlap-grouped-vectors}
 S=(s_1,s_2,s_3,\sqrt Q),\qquad
 T=(t_1,t_2,t_3,\norm{t_{\ge4}}_2).
\end{equation}
Both vectors have norm one. The overlap of two pure states is at most
the scalar product of their ordered Schmidt coefficients, and
Cauchy--Schwarz bounds the sum over the remaining coefficients:
\begin{equation}\label{eq:overlap-grouped-fidelity}
 \sqrt{1-f}=|\langle\psi|\chi\rangle|
 \le\sum_i s_it_i\le S\cdot T.
\end{equation}
Define the real number
\[
 \ell=\frac{s_1-s_2-s_3-\sqrt Q}{2}.
\]
Our goal is to prove $\ell\le\sqrt f$.
This already holds if $\ell\le0$, so assume $\ell>0$.
Take the unit vector $e=(1,-1,-1,-1)/2$.
Then $e\cdot S=\ell$, while~\eqref{eq:overlap-endpoint-halfspace}
gives $e\cdot T\le0$.
Here $\ell e=(\ell/2,-\ell/2,-\ell/2,-\ell/2)$ is simply
the vector obtained by multiplying each coordinate of $e$ by
the number $\ell$. Subtracting this vector from $S$ is useful because
\[
 (S-\ell e)\cdot T=S\cdot T-\ell(e\cdot T)\ge S\cdot T.
\]
Thus Cauchy--Schwarz gives
\begin{equation}\label{eq:overlap-halfspace-geometry}
 \sqrt{1-f}\le S\cdot T
 \le(S-\ell e)\cdot T
 \le\norm{S-\ell e}_2\norm T_2
 =\sqrt{1-\ell^2}.
\end{equation}
The last equality uses
$\norm{S-\ell e}_2^2=\norm S_2^2-2\ell(e\cdot S)+\ell^2\norm e_2^2
=1-\ell^2$.
Squaring the two ends of~\eqref{eq:overlap-halfspace-geometry}
gives $\ell^2\le f$, hence $2\ell\le2\sqrt f$.
This proves~\eqref{eq:overlap-direct-tail} in both cases.
Appendix~\ref{app:overlap-transfer} verifies the projected marginals
and proves the Schmidt-overlap inequality by finite sums.

\end{proof}

\subsection{The entropy minimum}

The next lemma gives the entropy minimum under the four-probability
constraint obtained from Lemma~\ref{lem:overlap-tail}. For the Werner
state we will set $c=2\sqrt f$; the condition $c\ge1/4$ is the source
of the lower endpoint $f\ge1/64$. Grouping the tail lets this
four-entry minimum control purifications of arbitrary finite
ancillary dimensions, as shown in Proposition~\ref{prop:overlap-optimum}.

\begin{lemma}
\label{lem:overlap-entropy}
Let $1/4\le c\le1$, and set
\begin{equation}\label{eq:overlap-canonical-scalar}
 A_c=\frac{c+\sqrt{12-3c^2}}4,\qquad B_c=\frac{A_c-c}{3}.
\end{equation}
For every probability vector $q=(q_0,q_1,q_2,q_3)$ with $q_0$ a
largest entry and
\begin{equation}\label{eq:overlap-fourconstraint}
 \sqrt{q_0}-\sqrt{q_1}-\sqrt{q_2}-\sqrt{q_3}\le c,
\end{equation}
we have
\begin{equation}\label{eq:overlap-scalar}
 H(q)\ge H(A_c^2,B_c^2,B_c^2,B_c^2).
\end{equation}
The displayed spectrum attains the minimum. For $c<1$, it is the
unique minimizing spectrum with the specified leading entry, and
its entropy is strictly less than one bit.
\end{lemma}

\begin{proof}
The case $c=1$ has the pure minimizer $(1,0,0,0)$. Suppose $c<1$.
Continuity and compactness ensure that a minimum exists. The fixed
vector $(5/6,1/18,1/18,1/18)$ is feasible for every $c\ge1/4$ and has
entropy below one bit. A minimizing vector must therefore have
$q_0>1/2$. If its square-root constraint were strict, transferring
probability from the tail to this largest entry would decrease
entropy while remaining feasible. Thus the constraint is active.
Writing the amplitudes as $a=\sqrt{q_0}$ and $b_i=\sqrt{q_i}$, we get
\begin{equation}\label{eq:overlap-active-amplitudes}
 a^2+b_1^2+b_2^2+b_3^2=1,\qquad a-b_1-b_2-b_3=c.
\end{equation}

These two equalities describe the tradeoff between the leading
amplitude and the tail. If a tail entry vanishes, adding a small
amplitude changes its contribution to the square-root sum to first
order, while its probability changes only to second order. Adjusting
the leading amplitude and one positive tail to preserve both
equalities then decreases entropy. Hence every minimizing tail entry
is positive.

If the tails are unequal, vary the smallest tail amplitude and
compensate with the largest tail and the leading amplitude. Entropy
must be stationary along this feasible variation. The resulting
logarithmic relation between the amplitude ratios, combined with
normalization, forces $c<1/4$. An elementary comparison of tangents
to a convex and a concave function proves this estimate in
Appendix~\ref{app:overlap-stationary-proof}. It contradicts the
assumed range of $c$, so all three tails coincide. Solving
$a^2+3b^2=1$ and $a-3b=c$ yields $a=A_c$ and $b=B_c$.
Appendix~\ref{app:overlap-entropy-proof} gives the exact feasible
variations, entropy derivatives, and the trial-vector inequalities.
\end{proof}

\subsection{The single-copy optimum}

Lemmas~\ref{lem:overlap-tail} and~\ref{lem:overlap-entropy} together
give the lower bound needed to prove that the canonical purification
is optimal. For an arbitrary purification, we combine all Schmidt
probabilities from the fourth onward into one entry. This cannot
increase entropy and reduces the comparison to four probabilities,
regardless of the ancillary dimensions.

\begin{proposition}
\label{prop:overlap-optimum}
For every $1/64\le f\le1/4$, optimizing over arbitrary finite-dimensional
purifying systems gives
\begin{equation}\label{eq:overlap-optimum}
 \EP(W(f))=H(a_f^2,b_f^2,b_f^2,b_f^2)=C(f).
\end{equation}
The canonical purification attains this value.
\end{proposition}

\begin{proof}
The canonical purification supplies $\EP(W(f))\le C(f)$.
At $f=1/4$ it is the product vector
$\ket{\Phi^+}_{AA'}\ket{\Phi^+}_{BB'}$, so its zero cost is optimal.
Assume $1/64\le f<1/4$. Set $c=2\sqrt f\in[1/4,1)$.
The explicit formulas give
$A_c=a_f$ and $B_c=b_f$, so Lemma~\ref{lem:overlap-entropy} also
shows $C(f)<1$.

Take an arbitrary purification and let $p_1\ge p_2\ge\cdots$ be its
Schmidt probabilities across $AA':BB'$. If $p_1\le1/2$, each
nonzero $p_i$ satisfies $-\log_2p_i\ge1$. Therefore
$H(p)\ge1>C(f)$.

If $p_1>1/2$, set $Q=\sum_{i\ge4}p_i$. Then $p_1$ is the largest
entry of the probability vector $q=(p_1,p_2,p_3,Q)$, and
Lemma~\ref{lem:overlap-tail} gives
$\sqrt{p_1}-\sqrt{p_2}-\sqrt{p_3}-\sqrt Q\le c$.
Grouping the tail cannot increase entropy: explicitly,
$H(p)=H(q)+QH((p_i/Q)_{i\ge4})\ge H(q)$, with the second term
defined to be zero when $Q=0$. Applying
Lemma~\ref{lem:overlap-entropy} thus gives
\begin{equation}\label{eq:overlap-completed-bound}
 H(p)\ge H(p_1,p_2,p_3,Q)
 \ge H(A_c^2,B_c^2,B_c^2,B_c^2)=C(f).
\end{equation}
Every purification therefore costs at least $C(f)$. Taking the
infimum and combining with the canonical upper bound proves the
claim, including attainment.
\end{proof}

The von Neumann lower bound follows from the joint
tail constraint in Lemma~\ref{lem:overlap-tail} and the entropy
minimum in Lemma~\ref{lem:overlap-entropy}. Their proofs, including
all auxiliary scalar estimates, are given in
Appendix~\ref{app:overlap}. 

\section{Discussion}\label{sec:discussion}

Our results determine the entanglement of purification of two-qubit
Werner states throughout $1/64\le f\le1/4$. In this interval, the
canonical purification is optimal over all purifications, without
symmetry or ancillary-dimension restrictions, and its entanglement
therefore gives both the entanglement of purification and the reflected
entropy. The lower bound follows by controlling the largest Schmidt
probability jointly with the tail and reducing the entropy comparison
to four probabilities. The minimizing spectrum then has three equal
smaller probabilities, as in the canonical purification.

The comparison with the regularized quantity shows how this exact
single-copy value relates to the cost of collective preparation.
The ensemble inequality of \cite{ChenWinter} gives a cost strictly
below the canonical value for $0<f<f_\ast$; the same auxiliary state
$W(f_\ast)$ can be used throughout this range. Combining this bound
with canonical optimality establishes nonadditivity for
$1/64\le f<f_\ast$. Thus, even where the canonical purification
minimizes the entanglement for one copy, collective preparation can
strictly reduce the cost per copy.

This collective advantage occurs entirely among separable states.
These Werner states nevertheless have correlations that are not
diagonal in a local product basis: for $f<1/4$, their Pauli correlation
matrix is a nonzero multiple of $\id_3$, whereas a classical two-qubit
state with maximally mixed marginals has a correlation matrix of rank
at most one. Our result therefore complements the R\'enyi
nonadditivity result for classical states in \cite{NegariKhanian},
while leaving the corresponding von Neumann additivity question open.

The AI-generated proof of nonadditivity in Ref.~\cite{Krohn} is incorrect. Appendix~\ref{app:error} gives an
explicit counterexample to Lemma~2 at its stated parameter
$f=1/200$. That lemma is false, invalidating the derivation of the
single-copy lower bound and the claimed proof of nonadditivity. The
error is a mismatch of weights in an application of Cauchy--Schwarz.

Several questions remain open. Our argument guarantees a finite-copy
violation of additivity without identifying the required number of
copies. It remains to determine whether two copies suffice and to
find the exact regularized value. The endpoints of the proven interval
reflect the bounds used here: $1/64$ comes from a sufficient estimate
in the entropy minimization, while $f_\ast$ marks the limit of the
ensemble comparison. Neither is claimed to be a sharp boundary for
nonadditivity, and canonical optimality may extend below $1/64$. It would also be interesting to check which parts of these proofs extend to higher Rényi entropies and whether they can be used to prove nonadditivity there.

\emph{Acknowledgments.}
Z.B.K. acknowledges support from the Ada Lovelace Postdoctoral Fellowship at Perimeter Institute for Theoretical Physics. Research at Perimeter Institute is supported in part by the Government of Canada through the Department of Innovation, Science and Economic Development and by the Province of Ontario through the Ministry of Colleges and Universities. The authors used a centaur (human–AI) workflow in the preparation of this work. The conceptual ideas and scientific contributions are those of the authors. Agentic AI tools assisted with detailed calculations, all of which were independently verified by the authors. The main proofs of canonical optimality and non-additivity presented in this paper were also formalized and verified in Lean. The formalization was carried out by AI, while the authors checked the structure of the theorems, as well as the axioms and assumptions used in the formal proofs.

\appendix
\numberwithin{equation}{section}
\section{Werner-state identities}\label{app:werner-identities}

\subsection{A separable decomposition}\label{app:separable}

The state $W(1/64)$ is separable. This can be seen without a separability
test. Let $\sigma_j$, $j\in\{x,y,z\}$, be the Pauli matrices and set
$\rho_{j,s}=(\id+s\sigma_j)/2$ for $s\in\{-1,1\}$. Expanding the products gives
\begin{equation}\label{eq:separable}
W(f)=\sum_{j\in\{x,y,z\}}\sum_{s=\pm1}
\left[
\frac{1-2f}{6}\rho_{j,s}\otimes\rho_{j,s}
+\frac f3\rho_{j,s}\otimes\rho_{j,-s}
\right],\qquad 0\le f\le\frac12.
\end{equation}
All coefficients are nonnegative and sum to one. At $f=1/64$, the six
equal-sign products have weight $31/192$ each and the six opposite-sign
products have weight $1/192$ each. Thus the example is a mixture of twelve
pure product states.

\subsection{Bell-state recoupling}\label{app:bell}

The Bell vectors are real in the computational basis, so vectorizing
$\sqrt{W(f)}$ gives the canonical purification
\begin{equation}\label{eq:bellpurif}
\ket{\psi_f}
=\sqrt f\,\ket{\Psi^-}_{AB}\ket{\Psi^-}_{A'B'}
+\sqrt{\frac{1-f}{3}}
\sum_{\beta\in\{\Phi^+,\Phi^-,\Psi^+\}}
\ket\beta_{AB}\ket\beta_{A'B'}.
\end{equation}
The expression above is a Schmidt decomposition across
$AB:A'B'$.
Its Schmidt amplitudes across $AA':BB'$ are $(a_f,b_f,b_f,b_f)$, with
\begin{equation}\label{eq:canonical-amplitudes}
a_f=\frac{\sqrt f+\sqrt{3(1-f)}}2,
\qquad
b_f=\frac{\sqrt{(1-f)/3}-\sqrt f}{2}.
\end{equation}
For $0\le f\le1/4$, these amplitudes satisfy
\begin{equation}\label{eq:canonical-relations}
a_f^2+3b_f^2=1,\qquad a_f-3b_f=2\sqrt f,\qquad a_f\ge b_f\ge0.
\end{equation}

For the purification~\eqref{eq:bellpurif}, recall that
$a=\sqrt f$, $b=\sqrt{(1-f)/3}$, and
\begin{equation*}
\ket{\Phi^\pm}=\frac{\ket{00}\pm\ket{11}}{\sqrt2},\qquad
\ket{\Psi^+}=\frac{\ket{01}+\ket{10}}{\sqrt2}.
\end{equation*}
Expanding the Bell states in the computational basis and regrouping
as $AA':BB'$ gives
\begin{align}\label{eq:recouple}
\ket{\psi_f}
={}&\frac{a+3b}{2}\ket{\Phi^+}_{AA'}\ket{\Phi^+}_{BB'}
+\frac{b-a}{2}\ket{\Phi^-}_{AA'}\ket{\Phi^-}_{BB'}\notag\\
&+\frac{b-a}{2}\ket{\Psi^+}_{AA'}\ket{\Psi^+}_{BB'}
+\frac{a-b}{2}\ket{\Psi^-}_{AA'}\ket{\Psi^-}_{BB'}.
\end{align}
For $0\le f\le1/4$, the absolute values of the coefficients are
$(a_f,b_f,b_f,b_f)$ from~\eqref{eq:canonical-amplitudes}. The sign of
the last Bell vector can be absorbed locally, so this is a Schmidt
decomposition. At $t_0=4/247$, the squared amplitudes are
\begin{equation}\label{eq:rvalue}
q(t_0)=\frac{841}{988},\qquad b_{t_0}^2=\frac{49}{988}.
\end{equation}
At $f=1/64$, they give
\begin{equation}\label{eq:point-value}
C(1/64)=H\!\left(\frac{95+3\sqrt{21}}{128},
\frac{11-\sqrt{21}}{128},\frac{11-\sqrt{21}}{128},
\frac{11-\sqrt{21}}{128}\right)\approx0.849253.
\end{equation}
These identities verify
the canonical entropies used in Propositions~\ref{prop:upper}
and~\ref{prop:overlap-optimum}.

\subsection{The symmetric endpoint}\label{app:wzero}

For example, the separable state $W(0)$ has $\EP(W(0))=1$.
Indeed, every purification is invariant under exchange of the physical
qubits. Thus $S(AA')=S(BA')$, and weak monotonicity gives
\begin{equation}\label{eq:w0}
2S(AA')=S(AA')+S(BA')\ge S(A)+S(B)=2.
\end{equation}
The reverse bound follows by assigning the full purifying system to
Bob, which gives $\EP\le S(A)=1$. This is the symmetric-support
argument of  \cite{ChristandlWinter}.

\section{Regularization bound}\label{app:arithmetic}

This appendix supplies the analytical proof of
Proposition~\ref{prop:upper} and an exact positive regularization gap.
All strict comparisons follow from elementary function bounds or
displayed integer inequalities. No numerical search or interval
subdivision is required.

\subsection{The derivative test and an exact gap}\label{app:upper-proof}

We compute the derivative used in Section~\ref{sec:upper}, determine
where it is negative. For the calculation, use the parameter $t=f/(1-p)$ and
write $U_f(t)=V_f(1-f/t)$, with $f\le t\le1/4$.

\paragraph{The initial derivative.}
The state entropy satisfies
\[
 S(W(t))-t\frac{d}{dt}S(W(t))=\log_2\frac{3}{1-t}.
\]
Using this identity to differentiate~\eqref{eq:variable-endpoint-cost}
gives
\begin{align}\label{eq:variable-endpoint-derivative}
 \frac{\partial U_f(t)}{\partial t}&=-\frac f{t^2}J(t),\notag\\
 J(t)&:=C(t)-tC'(t)+\log_2(1-t)-1.
\end{align}
To evaluate $J(t)$, write $z=\sqrt{3t/(1-t)}$. The canonical
probability and its derivative obey
\begin{equation*}
 q(t)=\frac{(3+z)^2}{4(3+z^2)},\qquad
 q(t)-tq'(t)=\frac{3+z}{4}.
\end{equation*}
Substituting $C(t)=\hbin(q(t))+(1-q(t))\log_2 3$ gives
\begin{align}\label{eq:J-binary}
 J(t)&=\log_2 6-\frac{1-z}{2}\log_2(1-z)
                  -\frac{3+z}{2}\log_2(3+z)\notag\\
     &=2\hbin\!\left(\frac{1-\sqrt{3t/(1-t)}}4\right)
          -\log_2\frac83.
\end{align}
Since $t=f/(1-p)$, the chain rule gives
$V_f'(0^+)=fU_f'(f)=-J(f)$,
which is~\eqref{eq:initial-endpoint-derivative}.

\paragraph{The sign and the threshold.}
As $t$ increases from $0$ to $1/4$, the binary-entropy argument in
\eqref{eq:J-binary} decreases strictly from $1/4$ to $0$.
Since $\hbin$ is strictly increasing on $[0,1/2]$, $J$ is strictly
decreasing. Its continuous endpoint values are
$J(0)=\log_2(2/\sqrt3)>0$ and $J(1/4)=-\log_2(8/3)<0$.
Thus it has exactly one zero, defining $f_\ast$ by
\begin{equation}\label{eq:strict-bound-threshold}
 2\hbin\!\left(\frac{1-\sqrt{3f_\ast/(1-f_\ast)}}4\right)
 =\log_2\frac83,\qquad f_\ast\approx0.017086.
\end{equation}
Consequently, $V_f'(0^+)<0$ precisely for $0<f<f_\ast$.
Figure~\ref{fig:initial-slope} illustrates this sign change.
For every such $f$, \eqref{eq:variable-endpoint-derivative} also shows
that $U_f(t)$ decreases up to $t=f_\ast$ and increases thereafter.
Thus $t=f_\ast$ minimizes this preparation bound for every
$0<f<f_\ast$.

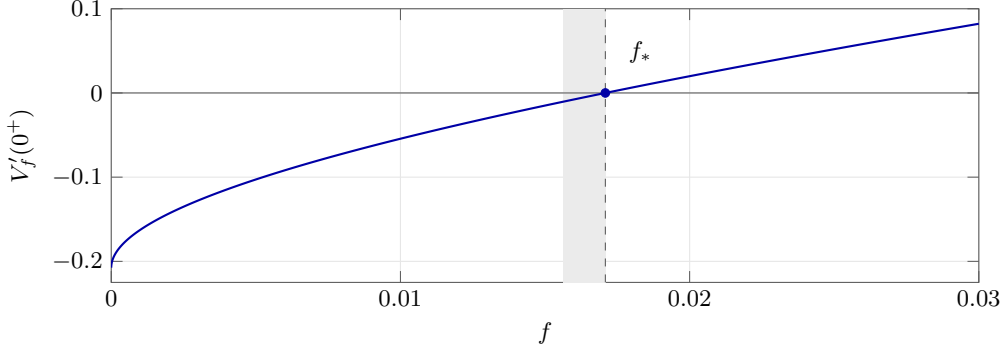
\begin{figure}[tb]
\centering
\begin{tikzpicture}
\begin{axis}[
 width=0.82\linewidth,height=5.2cm,
 xmin=0,xmax=0.03,ymin=-0.225,ymax=0.10,
 xlabel={$f$},ylabel={$V_f'(0^+)$},
 xtick={0,0.01,0.02,0.03},ytick={-0.2,-0.1,0,0.1},
 scaled ticks=false,tick label style={font=\small},
 xticklabel style={/pgf/number format/fixed,/pgf/number format/precision=2},
 axis line style={black!65},tick style={black!65},
 grid=major,grid style={black!10},
]
\path[fill=black!8] (axis cs:0.015625,-0.225)
 rectangle (axis cs:0.01708595559070344,0.10);
\addplot[black!65,thin] coordinates {(0,0) (0.03,0)};
\addplot[blue!65!black,thick] coordinates {
(0,-0.207518749639)
(9.25925925926e-07,-0.206196610619)
(3.7037037037e-06,-0.204871793279)
(8.33333333333e-06,-0.203544290921)
(1.48148148148e-05,-0.202214096803)
(2.31481481481e-05,-0.200881204143)
(3.33333333333e-05,-0.199545606116)
(4.53703703704e-05,-0.198207295854)
(5.92592592593e-05,-0.196866266448)
(7.5e-05,-0.195522510946)
(9.25925925926e-05,-0.194176022354)
(0.000112037037037,-0.192826793632)
(0.000133333333333,-0.191474817701)
(0.000156481481481,-0.190120087435)
(0.000181481481481,-0.188762595667)
(0.000208333333333,-0.187402335185)
(0.000237037037037,-0.186039298732)
(0.000267592592593,-0.184673479008)
(0.0003,-0.183304868669)
(0.000334259259259,-0.181933460324)
(0.00037037037037,-0.180559246541)
(0.000408333333333,-0.179182219838)
(0.000448148148148,-0.177802372691)
(0.000489814814815,-0.17641969753)
(0.000533333333333,-0.175034186738)
(0.000578703703704,-0.173645832652)
(0.000625925925926,-0.172254627564)
(0.000675,-0.170860563718)
(0.000725925925926,-0.169463633313)
(0.000778703703704,-0.168063828498)
(0.000833333333333,-0.166661141378)
(0.000889814814815,-0.165255564008)
(0.000948148148148,-0.163847088397)
(0.00100833333333,-0.162435706505)
(0.00107037037037,-0.161021410243)
(0.00113425925926,-0.159604191475)
(0.0012,-0.158184042016)
(0.00126759259259,-0.156760953631)
(0.00133703703704,-0.155334918037)
(0.00140833333333,-0.153905926899)
(0.00148148148148,-0.152473971834)
(0.00155648148148,-0.15103904441)
(0.00163333333333,-0.149601136142)
(0.00171203703704,-0.148160238495)
(0.00179259259259,-0.146716342885)
(0.001875,-0.145269440674)
(0.00195925925926,-0.143819523174)
(0.00204537037037,-0.142366581646)
(0.00213333333333,-0.140910607296)
(0.00222314814815,-0.139451591281)
(0.00231481481481,-0.137989524703)
(0.00240833333333,-0.136524398612)
(0.0025037037037,-0.135056204004)
(0.00260092592593,-0.133584931822)
(0.0027,-0.132110572955)
(0.00280092592593,-0.130633118237)
(0.0029037037037,-0.12915255845)
(0.00300833333333,-0.127668884318)
(0.00311481481481,-0.126182086511)
(0.00322314814815,-0.124692155644)
(0.00333333333333,-0.123199082276)
(0.00344537037037,-0.12170285691)
(0.00355925925926,-0.120203469992)
(0.003675,-0.118700911912)
(0.00379259259259,-0.117195173001)
(0.00391203703704,-0.115686243535)
(0.00403333333333,-0.114174113731)
(0.00415648148148,-0.112658773748)
(0.00428148148148,-0.111140213686)
(0.00440833333333,-0.109618423587)
(0.00453703703704,-0.108093393433)
(0.00466759259259,-0.106565113147)
(0.0048,-0.105033572591)
(0.00493425925926,-0.103498761568)
(0.00507037037037,-0.101960669819)
(0.00520833333333,-0.100419287027)
(0.00534814814815,-0.0988746028089)
(0.00548981481481,-0.0973266067233)
(0.00563333333333,-0.0957752882654)
(0.0057787037037,-0.0942206368677)
(0.00592592592593,-0.0926626419001)
(0.006075,-0.0911012926689)
(0.00622592592593,-0.0895365784167)
(0.0063787037037,-0.087968488322)
(0.00653333333333,-0.0863970114985)
(0.00668981481481,-0.0848221369952)
(0.00684814814815,-0.0832438537954)
(0.00700833333333,-0.0816621508169)
(0.00717037037037,-0.080077016911)
(0.00733425925926,-0.0784884408624)
(0.0075,-0.0768964113887)
(0.00766759259259,-0.07530091714)
(0.00783703703704,-0.0737019466984)
(0.00800833333333,-0.0720994885777)
(0.00818148148148,-0.0704935312225)
(0.00835648148148,-0.0688840630085)
(0.00853333333333,-0.0672710722413)
(0.00871203703704,-0.0656545471564)
(0.00889259259259,-0.0640344759187)
(0.009075,-0.0624108466218)
(0.00925925925926,-0.0607836472876)
(0.00944537037037,-0.059152865866)
(0.00963333333333,-0.0575184902344)
(0.00982314814815,-0.0558805081967)
(0.0100148148148,-0.0542389074838)
(0.0102083333333,-0.0525936757519)
(0.0104037037037,-0.0509448005832)
(0.0106009259259,-0.0492922694843)
(0.0108,-0.0476360698865)
(0.0110009259259,-0.045976189145)
(0.0112037037037,-0.0443126145381)
(0.0114083333333,-0.0426453332672)
(0.0116148148148,-0.0409743324559)
(0.0118231481481,-0.0392995991494)
(0.0120333333333,-0.0376211203143)
(0.0122453703704,-0.0359388828379)
(0.0124592592593,-0.0342528735275)
(0.012675,-0.0325630791099)
(0.0128925925926,-0.030869486231)
(0.013112037037,-0.0291720814551)
(0.0133333333333,-0.0274708512643)
(0.0135564814815,-0.025765782058)
(0.0137814814815,-0.0240568601522)
(0.0140083333333,-0.0223440717789)
(0.014237037037,-0.0206274030859)
(0.0144675925926,-0.0189068401355)
(0.0147,-0.0171823689044)
(0.0149342592593,-0.0154539752829)
(0.0151703703704,-0.0137216450744)
(0.0154083333333,-0.0119853639946)
(0.0156481481481,-0.010245117671)
(0.0158898148148,-0.00850089164205)
(0.0161333333333,-0.0067526713568)
(0.0163787037037,-0.00500044217404)
(0.0166259259259,-0.00324418936165)
(0.016875,-0.001483898096)
(0.0171259259259,0.000280446538789)
(0.0173787037037,0.00204885955149)
(0.0176333333333,0.0038213560445)
(0.0178898148148,0.0055979512145)
(0.0181481481481,0.00737866035318)
(0.0184083333333,0.00916349884798)
(0.0186703703704,0.0109524821827)
(0.0189342592593,0.0127456259385)
(0.0192,0.0145429457943)
(0.0194675925926,0.0163444575275)
(0.019737037037,0.0181501770154)
(0.0200083333333,0.019960120235)
(0.0202814814815,0.0217743032644)
(0.0205564814815,0.0235927422835)
(0.0208333333333,0.0254154535746)
(0.021112037037,0.0272424535232)
(0.0213925925926,0.029073758619)
(0.021675,0.0309093854567)
(0.0219592592593,0.0327493507364)
(0.0222453703704,0.0345936712651)
(0.0225333333333,0.036442363957)
(0.0228231481481,0.0382954458344)
(0.0231148148148,0.0401529340289)
(0.0234083333333,0.0420148457819)
(0.0237037037037,0.0438811984458)
(0.0240009259259,0.0457520094843)
(0.0243,0.0476272964742)
(0.0246009259259,0.0495070771053)
(0.0249037037037,0.0513913691822)
(0.0252083333333,0.0532801906245)
(0.0255148148148,0.0551735594682)
(0.0258231481481,0.0570714938665)
(0.0261333333333,0.0589740120908)
(0.0264453703704,0.0608811325314)
(0.0267592592593,0.0627928736988)
(0.027075,0.0647092542246)
(0.0273925925926,0.0666302928624)
(0.027712037037,0.0685560084888)
(0.0280333333333,0.0704864201047)
(0.0283564814815,0.0724215468358)
(0.0286814814815,0.0743614079341)
(0.0290083333333,0.0763060227789)
(0.029337037037,0.0782554108777)
(0.0296675925926,0.0802095918671)
(0.03,0.0821685855144)
};
\addplot[black!65,dashed] coordinates
 {(0.01708595559070344,-0.225) (0.01708595559070344,0.10)};
\addplot[only marks,mark=*,mark size=1.6pt,blue!65!black]
 coordinates {(0.01708595559070344,0)};
\node[anchor=south west,font=\small] at (axis cs:0.0176,0.026) {$f_\ast$};
\end{axis}
\end{tikzpicture}
\caption{The initial derivative~\eqref{eq:initial-endpoint-derivative}
is negative for $0<f<f_\ast$, so a small positive mixing probability
lowers the achievable rate below $C(f)$. The value at $f=0$ is the
continuous limit. In the shaded interval $[1/64,f_\ast)$, canonical
purification is also optimal for one copy, establishing nonadditivity.}
\label{fig:initial-slope}
\end{figure}

\section{Details of the overlap approach}\label{app:overlap}

\subsection{The canonical purification maximizes the largest eigenvalue}
\label{app:overlap-leading}

This result connects the Werner case with the largest-eigenvalue
intuition discussed at the beginning of Section~\ref{sec:overlap}.

For states $\rho,\sigma$ on the same system, the root fidelity is
\[
 \mathcal F(\rho,\sigma)
 =\Tr\sqrt{\sqrt\rho\,\sigma\sqrt\rho}
 =\norm{\sqrt\rho\sqrt\sigma}_1,
\]
where $\norm X_1=\Tr\sqrt{X^\dagger X}$ is the trace norm.
If $\ket\psi_{AA'BB'}$ purifies $\rho_{AB}$ and has largest Schmidt
probability $p_1$ across $AA':BB'$, then $\sqrt{p_1}$ is its maximum
absolute overlap with a normalized product vector
$\ket\alpha_{AA'}\ket\beta_{BB'}$.
This vector purifies $\tau_A\otimes\omega_B$, where
$\tau_A=\Tr_{A'}\proj\alpha$ and $\omega_B=\Tr_{B'}\proj\beta$.
The purification characterization of fidelity\cite{Uhlmann}
therefore gives
\begin{equation}\label{eq:overlap-product-reduction}
 \sqrt{p_1}\le
 \max_{\tau_A,\omega_B}\mathcal F(\rho_{AB},\tau_A\otimes\omega_B),
\end{equation}
where the maximum is over density operators on $A$ and $B$. This
bound holds for every finite choice of ancillary dimensions.
\par

\begin{appendixlemma}\label{lem:overlap-leading}
Let $0\le f\le1/4$, and let $\ket\psi_{AA'BB'}$ be any purification
of $W(f)$, with arbitrary finite ancillary dimensions. Let $p_1(\psi)$
be the largest eigenvalue of the reduced state $\Tr_{BB'}\proj\psi$.
Then
\begin{equation}\label{eq:overlap-leading-bound}
 p_1(\psi)\le a_f^2.
\end{equation}
The canonical purification attains equality. Thus it has the largest
possible leading Schmidt probability across $AA':BB'$ among all
purifications of $W(f)$.
\end{appendixlemma}

\begin{proof}
For normalized product vectors
$\ket\alpha_{AA'}\ket\beta_{BB'}$, Uhlmann's theorem\cite{Uhlmann}
gives
\begin{equation}\label{eq:overlap-uhlmann}
 \bigl|\bra\alpha\bra\beta\psi\rangle\bigr|
 \le\mathcal F\bigl(W(f),\tau_A\otimes\omega_B\bigr),
 \quad
 \tau_A=\Tr_{A'}\proj\alpha,\quad
 \omega_B=\Tr_{B'}\proj\beta.
\end{equation}
Maximizing over the product vectors recovers
\eqref{eq:overlap-product-reduction} for $\rho_{AB}=W(f)$.
We will show that the maximum on the right-hand side is
\begin{equation}\label{eq:overlap-product-maximum}
 \max_{\tau_A,\omega_B}\mathcal F\bigl(W(f),\tau_A\otimes\omega_B\bigr)
 =\frac{\Tr\sqrt{W(f)}}2=a_f.
\end{equation}
This is the auxiliary fidelity calculation used to bound $p_1$.

\paragraph{Reduction to commuting qubit states.}
Unitary invariance of root fidelity and
$(U\otimes U)W(f)(U\otimes U)^\dagger=W(f)$ allow us to diagonalize
$\tau$. In that basis put $Z=\operatorname{diag}(1,-1)$ and
$\mathcal D(\omega)=(\omega+Z\omega Z)/2$.
Since $Z\tau Z=\tau$, simultaneous conjugation by $Z\otimes Z$
leaves $W(f)$ fixed and gives
\begin{equation*}
 \mathcal F(W(f),\tau\otimes Z\omega Z)
 =\mathcal F(W(f),\tau\otimes\omega).
\end{equation*}
Concavity of root fidelity in its second argument then implies
\begin{equation}\label{eq:overlap-dephasing}
 \mathcal F(W(f),\tau\otimes\mathcal D(\omega))
 \ge\mathcal F(W(f),\tau\otimes\omega).
\end{equation}
For completeness, this concavity follows from the purification
characterization: fix a purification of the first state, choose
purifications of the two second states with nonnegative optimal
overlaps, and attach orthogonal flags with amplitudes
$\sqrt{\lambda}$ and $\sqrt{1-\lambda}$. Their coherent flagged sum
purifies the mixture, and its overlap with the fixed purification
tensored with the same flag superposition is the corresponding
weighted average. Uhlmann's theorem bounds that average by the root
fidelity with the mixture.

\paragraph{The two-dimensional block calculation.}
Write
\begin{equation}\label{eq:overlap-angle-parameters}
 \tau=\operatorname{diag}(\cos^2\theta,\sin^2\theta),\qquad
 \omega=\operatorname{diag}(\cos^2\phi,\sin^2\phi),
 \qquad 0\le\theta,\phi\le\pi/2,
\end{equation}
and put $u=\sqrt{(1-f)/3}$ and $v=\sqrt f$, so $u\ge v\ge0$.
In the ordered computational basis, the Werner state is
\begin{equation}\label{eq:overlap-werner-block}
 W(f)=
 \begin{pmatrix}
 u^2&0&0&0\\
 0&(u^2+v^2)/2&(u^2-v^2)/2&0\\
 0&(u^2-v^2)/2&(u^2+v^2)/2&0\\
 0&0&0&u^2
 \end{pmatrix}.
\end{equation}
The $00$ and $11$ blocks contribute
$u(\cos\theta\cos\phi+\sin\theta\sin\phi)$ to root fidelity.
For the middle block of
$\sqrt{\tau\otimes\omega}\,\allowbreak W(f)\allowbreak\sqrt{\tau\otimes\omega}$,
set $m=\cos\theta\sin\phi$ and $n=\sin\theta\cos\phi$.
This block has trace $(u^2+v^2)(m^2+n^2)/2$ and determinant
$u^2v^2m^2n^2$. For every positive semidefinite $2\times2$ matrix $K$,
\begin{equation}\label{eq:overlap-two-eigenvalues}
 (\Tr\sqrt K)^2=\Tr K+2\sqrt{\det K},
\end{equation}
as follows by writing its two eigenvalues. Hence, with
$X=\cos(\theta-\phi)$ and $Y=\sin(\theta+\phi)$,
\begin{equation}\label{eq:overlap-explicit-fidelity}
 \mathcal F(W(f),\tau\otimes\omega)
 =uX+\frac12\sqrt{(u+v)^2Y^2+(u-v)^2(1-X^2)}.
\end{equation}
Here $0\le X,Y\le1$. Only elementary inequalities are now needed:
\begin{align}\label{eq:overlap-elementary-maximum}
 \sqrt{(u+v)^2Y^2+(u-v)^2(1-X^2)}
 &\le(u+v)\sqrt{2-X^2}\notag\\
 &\le(u+v)(2-X).
\end{align}
The first line uses $Y^2\le1$ and $(u-v)^2\le(u+v)^2$.
The second follows from
$(2-X)^2-(2-X^2)=2(1-X)^2\ge0$. Consequently,
\begin{equation}\label{eq:overlap-final-fidelity-bound}
 \mathcal F(W(f),\tau\otimes\omega)
 \le(u+v)+\frac{u-v}{2}X
 \le\frac{3u+v}{2}=a_f.
\end{equation}
Equality is attained at $\theta=\phi=\pi/4$, corresponding to
$\tau=\omega=\id/2$.

Since $\Tr\sqrt{W(f)}=3u+v$, this proves
\eqref{eq:overlap-product-maximum}. Combining the fidelity bound
with~\eqref{eq:overlap-product-reduction} gives $p_1\le a_f^2$.
The canonical spectrum has largest entry $a_f^2$, so equality is
attained, proving the lemma.
\end{proof}

\subsection{Why the largest eigenvalue does not determine the entropy}
\label{app:overlap-majorization}

For the classical states in  \cite{NegariKhanian}, the optimal
spectrum has at most two nonzero entries. Bounding its leading entry
therefore proves full majorization and the entropy bound. A Werner
canonical spectrum instead has four nonzero entries when $f<1/4$.
Indeed, assigning the entire purifying register to $B'$ gives a valid
purification with spectrum $(1/2,1/2)$ across $A:BB'$, because
$W(f)_A=\id/2$. Its first two probabilities sum to one, whereas
\begin{equation}\label{eq:overlap-majorization-obstruction}
 a_f^2+b_f^2=1-2b_f^2<1\qquad(f<1/4).
\end{equation}
Thus a universal majorization argument cannot complete the Werner proof.

The required comparison can be expressed exactly in terms of the
entropy of the tail. Fix $1/64\le f<1/4$ and a competing purification
with $x=p_1>1/2$. Lemma~\ref{lem:overlap-leading} gives $x\le a_f^2$.
Normalize its remaining probabilities by
$r_i=p_{i+1}/(1-x)$. Then
\begin{align}\label{eq:overlap-entropy-balance}
 H(p)-C(f)
 &=\hbin(x)-\hbin(a_f^2)\notag\\
 &\quad +(1-x)H(r)-(1-a_f^2)\log_2 3.
\end{align}
The first difference is nonnegative, because binary entropy decreases
on $[1/2,1]$. A smaller leading probability therefore increases the
binary part of the entropy, but may allow a more concentrated tail.
To prove optimality, we must show that this concentration cannot save
more entropy than the binary part gains. Equivalently, the remaining
comparison is
\begin{equation}\label{eq:overlap-tail-target}
 (1-x)H(r)\ge C(f)-\hbin(x).
\end{equation}
When $\hbin(x)\ge C(f)$, nonnegativity of $H(r)$ already suffices.
The joint constraint in Lemma~\ref{lem:overlap-tail} supplies the
additional information needed for this comparison. Its combination
with the entropy minimum in Lemma~\ref{lem:overlap-entropy} proves
Proposition~\ref{prop:overlap-optimum}, without requiring the
largest-eigenvalue bound or the entropy-balance decomposition above.

\subsection{A spectral half-space from the two-qubit extension criterion}
\label{app:overlap-halfspace}

We derive the endpoint constraint used in
Lemma~\ref{lem:overlap-tail} directly from  \cite{CJKLZ}.
Let $\sigma_{XY}$ have a symmetric extension in the qubit $Y$, with
$\sigma_Y=\id/2$, and write its ordered eigenvalues as
$\lambda_1\ge\lambda_2\ge\cdots$. Put
$R=\sum_{j\ge4}\lambda_j$. We will prove
\begin{equation}\label{eq:overlap-endpoint-spectrum}
 \sqrt{\lambda_1}\le\sqrt{\lambda_2}+\sqrt{\lambda_3}+\sqrt R.
\end{equation}
If $\lambda_1\le1/2$, the square of the right-hand side is at
least $1-\lambda_1\ge\lambda_1$. Assume therefore that
$\alpha:=\lambda_1>1/2$.

\paragraph{Compression to two qubits.}
Choose a unit eigenvector $\ket v$ for $\alpha$, and let $P$ project
onto its Schmidt support on $X$ across $X:Y$. This support has
dimension at most two. It cannot have dimension one: then
$\sigma_Y\succeq\alpha\proj y$ for a unit qubit vector $\ket y$,
contradicting $\sigma_Y=\id/2$. Thus
\begin{equation}\label{eq:overlap-compressed-state}
 \omega=(P\otimes\id_Y)\sigma(P\otimes\id_Y)
\end{equation}
is a subnormalized two-qubit state. Compressing a symmetric extension
on $X$ preserves that extension, and the eigenvalues of $\omega$ are
$\alpha,r_2,r_3,r_4$, with
\begin{equation}\label{eq:overlap-compressed-eigenvalues}
 0\le r_j\le\lambda_j\quad(j=2,3,4).
\end{equation}
These inequalities follow by restricting the subspaces in the
variational characterization of eigenvalues. Also,
$0\preceq\omega_Y\preceq\id/2$: the difference from $\sigma_Y$
is the partial trace of the positive complementary diagonal block.

\paragraph{The criterion and the scalar inequality.}
The necessary direction of the two-qubit extension criterion gives
\begin{equation}\label{eq:overlap-compressed-criterion}
 \Tr\omega^2-4\sqrt{\det\omega}\le\Tr\omega_Y^2.
\end{equation}
It applies to $\omega$ by homogeneity. Set $w=\Tr\omega$.
If $\mu,\nu\le1/2$ are the eigenvalues of $\omega_Y$, then
\begin{equation}\label{eq:overlap-compressed-purity}
 \Tr\omega_Y^2\le\frac14+\left(w-\frac12\right)^2,
\end{equation}
because the right-hand side minus $\mu^2+\nu^2$ equals
$2(1/2-\mu)(1/2-\nu)\ge0$.
Combining these two inequalities and expanding gives
$G_\alpha(r_2,r_3,r_4)\ge0$, where
\begin{equation}\label{eq:overlap-endpoint-polynomial}
 G_\alpha(x,y,z)
 =1-2\alpha+(4\alpha-2)(x+y+z)
   +4(xy+xz+yz)+8\sqrt{\alpha xyz}.
\end{equation}
For $\alpha>1/2$, this expression is nondecreasing in each
nonnegative variable. Therefore
\eqref{eq:overlap-compressed-eigenvalues} and $r_4\le\lambda_4\le R$
imply $G_\alpha(\lambda_2,\lambda_3,R)\ge0$.
Write $a=\sqrt\alpha$, $b=\sqrt{\lambda_2}$,
$c=\sqrt{\lambda_3}$ and $d=\sqrt R$. Using
$a^2+b^2+c^2+d^2=1$, direct multiplication yields
\begin{equation}\label{eq:overlap-endpoint-factorization}
 G_\alpha(\lambda_2,\lambda_3,R)
 =(a+b+c-d)(a+b-c+d)(a-b+c+d)(-a+b+c+d).
\end{equation}
The first three factors are positive because
$a^2>1/2>b^2,c^2,d^2$. The last factor is consequently
nonnegative, which proves~\eqref{eq:overlap-endpoint-spectrum}.
This is a complete derivation from the two-qubit criterion.

\subsection{Projection, overlap, and transfer of the tail constraint}
\label{app:overlap-transfer}

We provide the remaining details of Lemma~\ref{lem:overlap-tail}.
Let $V$ swap the physical qubits, and let
$P_-=(\id-V)/2=\proj{\Psi^-}$, with
$\Ps=(\id+V)/2$. For a state $\rho$ with both marginals $\id/2$
and $f=\Tr(P_-\rho)<1$, the one-dimensional antisymmetric subspace
gives $P_-\rho P_-=fP_-$. Cancelling the off-diagonal symmetric and
antisymmetric blocks gives
\begin{equation}\label{eq:overlap-swap-blocks}
 \rho+V\rho V=2\Ps\rho\Ps+2fP_-.
\end{equation}
After identifying the two qubit spaces in the computational basis,
the partial trace over $B$ of the left-hand side is
$\rho_A+\rho_B=\id$. Since $\Tr_B P_-=\id/2$, it follows that
\begin{equation}\label{eq:overlap-projected-marginal}
 \Tr_B(\Ps\rho\Ps)=\frac{1-f}{2}\id,
 \qquad
 \Tr_A(\Ps\rho\Ps)=\frac{1-f}{2}\id.
\end{equation}
Consequently the normalized vector $\ket\chi$ in
\eqref{eq:overlap-symmetric-projection} has physical marginals
$\id/2$ and satisfies $V\ket\chi=\ket\chi$.
The state $\Tr_{B'}\proj\chi$ on $A'AB$ is therefore a symmetric
extension of $\sigma_{A'A}=\Tr_{BB'}\proj\chi$.
Applying~\eqref{eq:overlap-endpoint-spectrum}, with $X=A'$ and $Y=A$,
to the eigenvalues $t_i^2$ of this marginal gives
\eqref{eq:overlap-endpoint-halfspace}.

\paragraph{The Schmidt-overlap bound by finite sums.}
Write the two Schmidt decompositions across $AA':BB'$ as
\begin{equation*}
 \ket\psi=\sum_i s_i\ket{u_i}\ket{v_i},\qquad
 \ket\chi=\sum_j t_j\ket{a_j}\ket{b_j},
\end{equation*}
with $s_i,t_j$ decreasing and nonnegative. Enlarge the ancillary
spaces by unused dimensions if needed, and pad by zeros to a common
finite length $n\ge4$. Put
$D_{ij}=|\langle u_i|a_j\rangle|\,|\langle v_i|b_j\rangle|$.
The triangle inequality gives
\begin{equation}\label{eq:overlap-double-sum}
 |\langle\psi|\chi\rangle|\le\sum_{i,j}s_it_jD_{ij}.
\end{equation}
Every row sum and every column sum of $D$ is at most one.
For example, Cauchy--Schwarz and orthonormality give
\begin{equation*}
 \sum_j D_{ij}\le
 \left(\sum_j|\langle u_i|a_j\rangle|^2\right)^{1/2}
 \left(\sum_j|\langle v_i|b_j\rangle|^2\right)^{1/2}\le1.
\end{equation*}
Set $z_i=\sum_jD_{ij}t_j$. For each $k$, the numbers
$c_j=\sum_{i=1}^kD_{ij}$ satisfy $0\le c_j\le1$ and
$\sum_jc_j\le k$. Hence
\begin{align}\label{eq:overlap-partial-sums}
 \sum_{i=1}^k z_i
 &=\sum_jc_jt_j\notag\\
 &\le\sum_{j=1}^k c_jt_j+t_k\sum_{j>k}c_j
 \le\sum_{j=1}^k\bigl[c_jt_j+(1-c_j)t_k\bigr]
 \le\sum_{j=1}^k t_j.
\end{align}
Taking $s_{n+1}=0$ and expanding finite sums now gives
\begin{equation}\label{eq:overlap-summation-by-parts}
 \sum_i s_i z_i
 =\sum_{k=1}^n(s_k-s_{k+1})\sum_{i=1}^kz_i
 \le\sum_{k=1}^n(s_k-s_{k+1})\sum_{i=1}^kt_i
 =\sum_i s_it_i.
\end{equation}
Together with~\eqref{eq:overlap-double-sum}, this proves
$|\langle\psi|\chi\rangle|\le\sum_i s_it_i$ using only
Cauchy--Schwarz, ordering, and finite sums.

\paragraph{The four-dimensional geometry.}
The vectors $S,T$ in~\eqref{eq:overlap-grouped-vectors} have norm one,
and Cauchy--Schwarz on the tails gives
$\sum_i s_it_i\le S\cdot T$.
For the unit vector $e=(1,-1,-1,-1)/2$,
\eqref{eq:overlap-endpoint-halfspace} gives $\eta=e\cdot T\le0$.
If $\ell=e\cdot S>0$, orthogonal decomposition yields
\begin{equation}\label{eq:overlap-complete-geometric-transfer}
 \sqrt{1-f}\le S\cdot T
 \le\ell\eta+\sqrt{1-\ell^2}\sqrt{1-\eta^2}
 \le\sqrt{1-\ell^2}.
\end{equation}
Both outer terms are nonnegative, so squaring gives $\ell^2\le f$.
If $\ell\le0$, the desired bound is immediate. Thus
$s_1-s_2-s_3-\norm{s_{\ge4}}_2=2\ell\le2\sqrt f$ in all cases.
This completes the proof of Lemma~\ref{lem:overlap-tail}.

\subsection{An elementary stationary-ratio estimate}
\label{app:overlap-stationary-proof}

The following estimate excludes unequal positive tails in the entropy
minimization. Its proof uses convexity and concavity to bound one
stationary amplitude ratio by an affine function of the other, then
checks a quadratic inequality with the exact tangent slope.

\begin{appendixlemma}\label{lem:overlap-stationary}
Let $u>t>1$ satisfy
\begin{equation}\label{eq:overlap-stationary-ratio}
\frac{u\ln u}{u+1}=\frac{t\ln t}{t-1}.
\end{equation}
Then
\begin{equation}\label{eq:overlap-stationary-bound}
\frac{u-t-2}{\sqrt{u^2+t^2+2}}<\frac14.
\end{equation}
\end{appendixlemma}

\begin{proof}
For $x>1$, define
\begin{equation}\label{eq:overlap-ratio-functions}
F(x)=\exp\!\left(\frac{x\ln x}{x+1}\right),\qquad
G(x)=\exp\!\left(\frac{x\ln x}{x-1}\right).
\end{equation}
Direct differentiation gives
\begin{align}\label{eq:overlap-ratio-curvature}
\frac{F''(x)}{F(x)}
&=\frac{(x+1)^2+x(\ln x)^2}{x(x+1)^4}>0,\notag\\
\frac{G''(x)}{G(x)}
&=\frac{x(\ln x)^2-(x-1)^2}{x(x-1)^4}\le0.
\end{align}
For the second sign, put $y=\sqrt x$. The function
$y^2-1-2y\ln y$ vanishes at $1$ and has derivative
$2(y-1-\ln y)\ge0$ for $y\ge1$. Thus
$\sqrt x\ln x\le x-1$. Also
$F'(x)/F(x)=(x+1+\ln x)/(x+1)^2>0$.

The hypothesis is $F(u)=G(t)$, and
$F(8)=G(4)=2^{8/3}$. Convexity of $F$ and concavity of $G$ give
\begin{equation}\label{eq:overlap-ratio-tangents}
F(8)+F'(8)(u-8)\le F(u)=G(t)
\le G(4)+G'(4)(t-4).
\end{equation}
After cancellation and division by $F'(8)>0$,
\begin{equation}\label{eq:overlap-ratio-slope}
u\le8+k(t-4),\qquad
k=\frac{G'(4)}{F'(8)}=\frac{3(3-2\ln2)}{3+\ln2}.
\end{equation}
Keep this exact value of $k$ and put $L(t)=8+k(t-4)$.
Direct expansion gives
\begin{align}\label{eq:overlap-ratio-polynomial}
&L(t)^2+t^2+2-16\bigl(L(t)-t-2\bigr)^2\notag\\
&\quad=(-15k^2+32k-15)(t-1)^2
 +\left[90\left(k-\frac43\right)^2+2\right](t-1)\notag\\
&\qquad\quad{}-135k^2+432k-333>0.
\end{align}
To check the signs analytically, the elementary bounds
$2/3<\ln2<25/36$ give $13/10<k<7/5$.
The quadratic and constant coefficients in~\eqref{eq:overlap-ratio-polynomial}
are concave functions of $k$; evaluating them at the endpoints of this
interval bounds them below by $2/5$ and $9/20$, respectively.
The linear coefficient is visibly positive. The logarithm bounds need
only the identity $\ln2=2\int_0^{1/3}(1-s^2)^{-1}\,ds$ and the
inequality $1<(1-s^2)^{-1}<1+9s^2/8$ for $0<s<1/3$.

For fixed $t>1$, the derivative with respect to $u>0$ of
$(u-t-2)/\sqrt{u^2+t^2+2}$ is
\begin{equation*}
\frac{t^2+2+u(t+2)}{(u^2+t^2+2)^{3/2}}>0.
\end{equation*}
Thus $u\le L(t)$ and~\eqref{eq:overlap-ratio-polynomial} imply
\begin{equation*}
\frac{u-t-2}{\sqrt{u^2+t^2+2}}
\le\frac{L(t)-t-2}{\sqrt{L(t)^2+t^2+2}}<\frac14.
\end{equation*}
This proves the lemma.
\end{proof}

\subsection{Proof of Lemma~\ref{lem:overlap-entropy}}
\label{app:overlap-entropy-proof}

\begin{proof}[Proof of Lemma~\ref{lem:overlap-entropy}]
Multiplying entropy by $\ln2$
does not change its minimizers, so use the natural-log entropy
$\mathcal H(q)=-\sum_iq_i\ln q_i$. The case $c=1$ is immediate:
$A_c=1$, $B_c=0$, and entropy is nonnegative. Assume $1/4\le c<1$.

\paragraph{Existence and an active constraint.}
The feasible set is closed in the probability simplex and therefore
compact. Entropy is continuous, including at zero probabilities, so a
minimum exists. The fixed vector
\begin{equation}\label{eq:overlap-trial-vector}
q_*=(5/6,1/18,1/18,1/18)
\end{equation}
is feasible for every $c\ge1/4$, because
\begin{equation*}
\sqrt{5/6}-\frac1{\sqrt2}
<\frac{11}{12}-\frac23=\frac14.
\end{equation*}
The two strict square-root estimates follow by squaring. Moreover,
\begin{equation}\label{eq:overlap-trial-entropy}
6\mathcal H(q_*)=\ln\!\left(18(6/5)^5\right)<\ln64=6\ln2,
\end{equation}
where $64\cdot5^5-18\cdot6^5=60032>0$.
If all probabilities are at most $1/2$, their entropy is at least one
bit. Thus every minimizer has $q_0>1/2$ and hence $q_0>q_i$ for
$i=1,2,3$. Also $q_0<1$, because the pure vector violates the
constraint when $c<1$.

If the constraint were strict at a minimizer, the perturbation
\begin{equation}\label{eq:overlap-concentrating-path}
q_0(\varepsilon)=q_0+\varepsilon(1-q_0),\qquad
q_i(\varepsilon)=(1-\varepsilon)q_i\quad(i=1,2,3)
\end{equation}
would remain feasible for sufficiently small positive $\varepsilon$.
Its entropy derivative at zero is
\begin{equation*}
\sum_{i:q_i>0}q_i\ln\frac{q_i}{q_0}<0.
\end{equation*}
At least one tail is positive, and each is smaller than $q_0$.
This contradicts minimality. Writing the amplitudes as
$(a,b_1,b_2,b_3)$, every minimizer therefore satisfies
\eqref{eq:overlap-active-amplitudes}, with $a>b_i$.

\paragraph{An explicit feasible variation.}
Choose two tail amplitudes $b>0$ and $d\ge0$, and hold the third
amplitude $e$ fixed. Replace $d$ by $d+\varepsilon$ and define
\begin{align}\label{eq:overlap-feasible-curve}
k(\varepsilon)&=a-b+\varepsilon,\notag\\
D(\varepsilon)&=2\bigl(a^2+b^2+d^2-(d+\varepsilon)^2\bigr)
-k(\varepsilon)^2,\notag\\
a(\varepsilon)&=\frac{k(\varepsilon)+\sqrt{D(\varepsilon)}}2,
\qquad
b(\varepsilon)=\frac{-k(\varepsilon)+\sqrt{D(\varepsilon)}}2.
\end{align}
At zero, $D(0)=(a+b)^2>0$ and these amplitudes equal $a,b$.
They remain positive nearby, and preserve both equalities exactly:
\begin{align*}
a(\varepsilon)-b(\varepsilon)-(d+\varepsilon)-e&=c,\\
a(\varepsilon)^2+b(\varepsilon)^2+(d+\varepsilon)^2+e^2&=1.
\end{align*}
Because $a$ is strictly larger than all the original tails, it remains
the leading amplitude for a sufficiently small variation. The derivatives
at zero are
\begin{equation}\label{eq:overlap-curve-derivatives}
a'(0)=\frac{b-d}{a+b},\qquad
b'(0)=-\frac{a+d}{a+b}.
\end{equation}
The amplitude entropy $h(x)=-x^2\ln(x^2)$ has derivative
$h'(x)=-4x\ln x-2x$ for $x>0$, and $h'(0)=0$ by continuous extension.

\paragraph{No tail entry vanishes.}
If a minimizing tail has $d=0$, choose another tail $b>0$; one exists
because $a<1$. For small $\varepsilon>0$,
\eqref{eq:overlap-feasible-curve} is feasible. Using~\eqref{eq:overlap-curve-derivatives}
and $h'(0)=0$, its entropy derivative is
\begin{equation}\label{eq:overlap-boundary-descent}
\left.\frac{\mathrm d}{\mathrm d\varepsilon}
\bigl[h(a(\varepsilon))+h(b(\varepsilon))+h(\varepsilon)+h(e)\bigr]
\right|_{\varepsilon=0^+}
=-\frac{4ab}{a+b}\ln\frac ab<0.
\end{equation}
This contradicts minimality, so all three tail amplitudes are positive.

\paragraph{Unequal positive tails are impossible.}
Suppose the largest tail $b$ exceeds the smallest tail $d>0$, with
$e$ the remaining tail. Then $a>b>d>0$ and $e\ge d$.
The curve~\eqref{eq:overlap-feasible-curve} is now feasible for both signs of
sufficiently small $\varepsilon$. Its derivative at a minimum must
vanish. Cancelling the normalization terms in the entropy derivative gives
\begin{equation}\label{eq:overlap-tail-stationary-identity}
ab\ln(a/b)=ad\ln(a/d)+bd\ln(b/d).
\end{equation}
Set $u=a/d$ and $t=b/d$, so $u>t>1$. Dividing the identity by $d^2$
and rearranging gives~\eqref{eq:overlap-stationary-ratio}. Normalization and
$e\ge d$ imply
\begin{equation*}
d^2(u^2+t^2+2)=a^2+b^2+2d^2\le1,
\qquad c=a-b-d-e\le d(u-t-2).
\end{equation*}
Since $c>0$, the numerator $u-t-2$ is positive. Lemma~\ref{lem:overlap-stationary}
therefore gives
\begin{equation}\label{eq:overlap-stationary-contradiction}
c=a-b-d-e\le d(u-t-2)
\le\frac{u-t-2}{\sqrt{u^2+t^2+2}}<\frac14,
\end{equation}
contradicting $c\ge1/4$. All three tail amplitudes must coincide.

Solving $a^2+3b^2=1$ and $a-3b=c$ with nonnegative amplitudes gives
\begin{equation*}
a=\frac{c+\sqrt{12-3c^2}}4=A_c,\qquad
b=\frac{A_c-c}{3}=B_c.
\end{equation*}
This vector is feasible and is the only possible minimizing spectrum
with the specified leading entry. Dividing by $\ln2$ proves
\eqref{eq:overlap-scalar}, its attainment, and the uniqueness assertion for $c<1$.
\end{proof}

\section{A counterexample to Lemma~2 of 
 \cite{Krohn}}\label{app:error}

We prove the counterexample to Lemma~2 of 
 \cite{Krohn}
and identify the error in the proof of Lemma~2 of 
 \cite{Krohn}.
Let
\begin{equation*}
\ket\psi=\sum_\alpha\sqrt{p_\alpha}\ket{x_\alpha}_{AA'}
\ket{y_\alpha}_{BB'}
\end{equation*}
be a purification of $W(f)$, where $0\le f<1/4$, written in an
ordered Schmidt decomposition across $AA':BB'$, with
$p_1\ge p_2\ge\cdots$.
Compress the physical Pauli operators to the two Schmidt bases:
\begin{equation}\label{eq:XY}
X^{(j)}_{\alpha\beta}
=\bra{x_\alpha}(\sigma_j\otimes\id)\ket{x_\beta},\qquad
Y^{(k)}_{\alpha\beta}
=\bra{y_\alpha}(\sigma_k\otimes\id)\ket{y_\beta}.
\end{equation}
The physical correlation matrix is
\begin{equation}\label{eq:corr}
\sum_{\alpha,\beta}\sqrt{p_\alpha p_\beta}
X^{(j)}_{\alpha\beta}Y^{(k)}_{\alpha\beta}
=\frac{1-4f}{3}\delta_{jk}.
\end{equation}
Let $R$ be the part of this sum for which at least one index exceeds two.
With $s=p_1+p_2$, $\varepsilon=1-s$, and $z=p_3$, the bound claimed
in Lemma~2 of that preprint, at its parameter $f=1/200$, is
\begin{equation}\label{eq:falsebound}
\norm R_{\mathrm{op}}\le\varepsilon+2sz.
\end{equation}

\begin{proposition}\label{prop:counterexample}
For every $0\le f<1/4$, the canonical purification of $W(f)$ admits
an ordered Schmidt decomposition across $AA':BB'$ for which
\begin{equation}\label{eq:counterexample}
\norm R_{\mathrm{op}}>\varepsilon+2sz.
\end{equation}
In particular, $f=1/200$ disproves Lemma~2 of version~1 of
 \cite{Krohn} within its stated domain, and $f=1/64$ gives a
counterexample at our main comparison point.
\end{proposition}

\begin{proof}
Put $d=\sqrt{(1-f)/3}-\sqrt f>0$. By~\eqref{eq:recouple}, the
canonical purification has Schmidt probabilities across $AA':BB'$
\begin{equation}\label{eq:counterweights}
p_1=P=a_f^2=1-3r,\qquad p_2=p_3=p_4=r=b_f^2=\frac{d^2}{4}.
\end{equation}
Choose $x_1=y_1=\Phi^+$ and $x_2=y_2=\Phi^-$.
This is an allowed choice within the degenerate Schmidt subspace.
The operator $\sigma_x\otimes\id$ maps both $\Phi^+$ and $\Phi^-$
outside their span. Hence all leading $2\times2$ matrix elements in
the $xx$ correlation vanish. Equation~\eqref{eq:corr} gives
\begin{equation}\label{eq:Rlower}
\norm R_{\mathrm{op}}\ge|R_{xx}|=\frac{1-4f}{3}
=d\left(\sqrt{\frac{1-f}{3}}+\sqrt f\right).
\end{equation}
Since $s=P+r=1-2r$, $\varepsilon=2r$, and $z=r$,
\begin{equation}\label{eq:Rupperfalse}
|R_{xx}|-(\varepsilon+2sz)
=d\left(\sqrt{\frac{1-f}{3}}+\sqrt f\right)
-d^2\left(1-\frac{d^2}{4}\right)
=2\sqrt f\,d+\frac{d^4}{4}>0.
\end{equation}
At the main comparison point $f=1/64$, this becomes
\begin{equation}\label{eq:counterexample-point}
|R_{xx}|=\frac5{16},\qquad
r=\frac{11-\sqrt{21}}{128}<\frac1{16},\qquad
\varepsilon+2sz=4r(1-r)<\frac14<\frac5{16}.
\end{equation}
Taking $f=1/200$ instead supplies the counterexample in the precise
domain of the cited lemma. The purification uses two-dimensional
ancillas, so it also lies inside the padded four-dimensional domain
used in  \cite{Krohn}.
\end{proof}

The incorrect step in the earlier proof occurs in Appendix~A.3 of
 \cite{Krohn}, immediately after the displayed formula for the
cross terms. For fixed Pauli directions, those terms are
\begin{equation}\label{eq:cross}
R_{\mathrm{cross}}
=2\operatorname{Re}\sum_{i\le2,\,\alpha\ge3}
\sqrt{p_i p_\alpha}\,X_{i\alpha}Y_{i\alpha}.
\end{equation}
The proof asserts $|R_{\mathrm{cross}}|\le2\sqrt{B_XB_Y}$ with
$B_X=\sum_{i\le2,\alpha\ge3}p_i p_\alpha|X_{i\alpha}|^2$ and the
same definition for $B_Y$. Cauchy--Schwarz does not give this inequality:
using these two squared norms would place the weight $p_i p_\alpha$,
rather than $\sqrt{p_i p_\alpha}$, in the sum being bounded.
A valid generic estimate is $|R_{\mathrm{cross}}|\le2\sqrt{s\varepsilon}$.
This follows by applying Cauchy--Schwarz to the leading and tail vectors,
whose squared norms are $s$ and $\varepsilon$, and using that the
product of the two physical Pauli observables has operator norm one.

\setlength{\bibsep}{4pt}
\begingroup
\raggedright
\bibliographystyle{unsrtnat}
\bibliography{references}
\endgroup
\end{document}